\documentclass{article}
\usepackage[T1]{fontenc}
\usepackage[dvipdfmx]{color}
\usepackage{amsmath}
\usepackage{amsthm}
\usepackage{amssymb}
\usepackage{geometry}

\usepackage[pdfusetitle,
 bookmarks=true,bookmarksnumbered=false,bookmarksopen=false,
 breaklinks=false,pdfborder={0 0 1},backref=false,colorlinks=false]
 {hyperref}

\makeatletter
\newcommand{\llladdress}[1]{
	\par {\raggedright #1
	\vspace{1.4em}
	\noindent\par}
}
\newenvironment{llllist}[1]
	{\begin{list}{}
		{\settowidth{\labelwidth}{#1}
		 \setlength{\leftmargin}{\labelwidth}
		 \addtolength{\leftmargin}{\labelsep}
		 }}
	{\end{list}}

\usepackage{stmaryrd}
\usepackage{cases}
\newcommand{\argmax}{\mathop{\rm arg~max}\limits}

\makeatother

\theoremstyle{plain}
\newtheorem{thm}{\protect\theoremname}
\newtheorem{lem}[thm]{\protect\lemmaname}
\newtheorem{cor}[thm]{\protect\corollaryname}

\providecommand{\corollaryname}{Corollary}
\providecommand{\lemmaname}{Lemma}
\providecommand{\theoremname}{Theorem}

\begin{document}
\title{Sufficient support size of measurements for quantum estimation}
\author{Koichi Yamagata}
\maketitle

\llladdress{Institute of Science and Engineering, Kanazawa University\\
Kanazawa, Ishikawa, 920-1192, Japan}

\global\long\def\E{\mathcal{E}}%
\global\long\def\S{\mathcal{S}}%
\global\long\def\R{\mathbb{R}}%
\global\long\def\C{\mathbb{C}}%
\global\long\def\N{\mathbb{N}}%
\global\long\def\Z{\mathbb{Z}}%
\global\long\def\D{\mathcal{D}}%
\global\long\def\M{\mathcal{M}}%
\global\long\def\X{\mathcal{X}}%
\global\long\def\F{\mathcal{F}}%
\global\long\def\B{\mathcal{B}}%
\global\long\def\T{\mathcal{T}}%
\global\long\def\P{\mathcal{P}}%
\global\long\def\H{\mathcal{H}}%
\global\long\def\Y{\mathcal{Y}}%
\global\long\def\A{\mathcal{A}}%
\global\long\def\bra#1{\left\langle #1\right|}%
\global\long\def\ket#1{\left|#1\right\rangle }%
\global\long\def\tr{{\rm tr}\,}%
\global\long\def\Tr{{\rm Tr}\,}%
\global\long\def\braket#1#2{\left\langle #1\mid#2\right\rangle }%
\global\long\def\V{\mathcal{V}}%
\global\long\def\re{{\rm Re}\,}%
\global\long\def\im{{\rm Im}\,}%
\global\long\def\id{{\rm id}}%
\global\long\def\ii{\mathrm{i}}%

\begin{abstract}
In quantum estimation for a $d$-parameter family of density operators
on a finite-dimensional Hilbert space $\H$, an estimator is specified
by a pair $\left(M,\hat{\theta}\right)$, where $M$ is a POVM with
a finite outcome set $\Omega$ and $\hat{\theta}:\Omega\to\R^{d}$
is a classical estimator map. Since the number of outcomes $\left|\Omega\right|$
is a priori unbounded, the space of admissible POVMs is vast, which
makes the search for optimal estimators difficult. In this paper,
for the minimization of the weighted trace of the mean squared error
among locally unbiased estimators, we prove that it suffices to consider
POVMs with at most $\left(\dim\H\right)^{2}+d(d+1)/2-1$ outcomes,
and that the optimization can be restricted to rank-one measurements.
 For Bayesian estimation with a general loss function, we show that
rank-one POVMs with at most $(\dim\mathcal{H})^{2}$ outcomes are
sufficient for the infimum value of the Bayes risk. Furthermore, when
the model admits a real sufficient subalgebra, we show that the $\left(\dim\H\right)^{2}$
term in the above support-size bounds can be reduced in both the locally
unbiased and Bayesian settings. These bounds substantially reduce
the search space for optimal measurements and justify restricting
numerical optimization to rank-one POVMs with finitely many outcomes.
\end{abstract}

\section{Introduction}

Quantum estimation theory concerns the design of measurements and
estimators for inferring parameters of quantum states, with applications
ranging from quantum sensing and tomography to benchmarking and calibration
of quantum devices. In finite dimensions, an estimator is specified
by a POVM together with a classical post-processing map. A basic obstacle
in optimizing measurement performance is that the number of POVM outcomes
is not fixed a priori: even for a fixed Hilbert space dimension, POVMs
with arbitrarily many outcomes exist. Consequently, numerical searches
for optimal measurements often proceed by restricting to POVMs with
a prescribed number of outcomes or a prescribed rank structure. Recent
works have developed practical numerical approaches for optimizing
POVMs in quantum parameter estimation \cite{kimizu,qest}. Without
a principled support-size guarantee, however, such finite-dimensional
searches do not by themselves certify global optimality.

The goal of this paper is to provide explicit finite-outcome guarantees
for optimal measurements in two standard quantum-estimation settings:
local estimation under the locally unbiased condition and Bayesian
estimation with a prior. We derive sufficient bounds on the number
of outcomes needed to preserve the optimal value of the corresponding
loss functions and show that the optimization may be restricted to
rank-one/extremal POVMs. Moreover, by exploiting model-adapted sufficient
operator subspaces\textemdash including real sufficient subalgebras
when available\textemdash we obtain sharper, structure-dependent bounds
that substantially reduce the effective search space for both analysis
and computation.

Let $\H$ be a finite-dimensional Hilbert space, and let $\S(\H)$
and $\B(\H)$ denote the set of density operators and linear operators
on $\H$, respectively. For $s\in\N$, let $\Omega_{s}:=\{1,2,\dots,s\}$,
and let $\M\left(\H,s\right)$ denote the set of POVMs $M=(M_{x})_{x\in\Omega_{s}}$
on $\H$ with outcome set $\Omega_{s}$, i.e., $M_{x}\geq0$ for all
$x\in\Omega_{s}$ and $\sum_{x}M_{x}=I$. Consider a $d$-parameter
family of density operators $\left\{ \rho_{\theta}\in\S(\H)\mid\theta\in\Theta\subset\R^{d}\right\} $.
An estimator is specified by a pair $\left(M,\hat{\theta}\right)$,
where $M\in\M(\H,s)$ for some $s\in\N$ and $\hat{\theta}:\Omega_{s}\to\R^{d}$
is a classical estimator map. We call $\left(M,\hat{\theta}\right)$
unbiased if, for all $\theta\in\Theta$,
\begin{equation}
E_{\theta}\left[M,\hat{\theta}^{i}\right]=\sum_{x\in\Omega_{s}}\hat{\theta}^{i}(x)\,\tr\rho_{\theta}M_{x}=\theta^{i},\qquad i=1,\dots,d.\label{eq:unbiased}
\end{equation}
We call an estimator $\left(M,\hat{\theta}\right)$ a locally unbiased
estimator (LUE) at a point $\theta_{0}\in\Theta$ \cite{holevo} if
the unbiasedness condition (\ref{eq:unbiased}) holds up to first
order around $\theta_{0}$, i.e.,
\begin{align}
\sum_{x\in\Omega_{s}}\hat{\theta}^{i}(x)\,\tr\rho_{\theta_{0}}M_{x} & =\theta^{i}_{0}\qquad(i=1,\dots,d),\label{eq:lue1}\\
\sum_{x\in\Omega_{s}}\hat{\theta}^{i}(x)\,\tr\partial_{j}\rho_{\theta_{0}}M_{x} & =\delta^{i}_{j}\qquad(i,j=1,\dots,d),\label{eq:lue2}
\end{align}
where $\partial_{j}\rho_{\theta_{0}}=\left.\frac{\partial}{\partial\theta^{j}}\rho_{\theta}\right|_{\theta=\theta_{0}}$.
For a locally unbiased estimator $\left(M,\hat{\theta}\right)$ at
$\theta_{0}$, define the mean squared error (MSE) matrix $V_{\theta_{0}}\left[M,\hat{\theta}\right]\in\R^{d\times d}$
by
\[
V_{\theta_{0}}\left[M,\hat{\theta}\right]^{ij}=\sum_{x}\left(\hat{\theta}^{i}(x)-\theta^{i}_{0}\right)\left(\hat{\theta}^{j}(x)-\theta^{j}_{0}\right)\tr\rho_{\theta_{0}}M_{x}.
\]
In quantum estimation theory, it is known $V_{\theta_{0}}\left[M,\hat{\theta}\right]$
satisfies the quantum Cram\'{e}r Rao inequality
\begin{equation}
V_{\theta_{0}}\left[M,\hat{\theta}\right]\geq J^{-1}_{\theta_{0}},\label{eq:SLDCRbound}
\end{equation}
where $J_{\theta_{0}}:=\left[\tr\partial_{i}\rho_{\theta_{0}}L_{j}\right]$
is the symmetric logarithmic derivative (SLD) Fisher information matrix.
The SLDs $L_{i}$ ($1\leq i\leq d$) are defined by
\[
\partial_{i}\rho_{\theta_{0}}=\frac{1}{2}\left(L_{i}\rho_{\theta_{0}}+\rho_{\theta_{0}}L_{i}\right).
\]
However, this matrix inequality cannot, in general, be saturated,
except when the SLDs commute, i.e., $\left[L_{i},L_{j}\right]=0$.
To avoid this difficulty, one often abandons minimizing $V_{\theta_{0}}\left[M,\hat{\theta}\right]$
itself and instead minimizes the weighted scalar criterion $\Tr WV_{\theta_{0}}\left[M,\hat{\theta}\right]$,
where $W$ is a given real positive definite $d\times d$ matrix (a
weight) \cite{holevo,helstrom}. In the local estimation setting,
we assume that the model is locally identifiable at $\theta_{0}$,
i.e., $\partial_{1}\rho_{\theta_{0}},\ldots,\partial_{d}\rho_{\theta_{0}}$
are linearly independent over $\mathbb{R}$. Under this condition,
the optimal value is finite; see Appendix \ref{sec:invertiblePOVM}.
Note that when $\rho_{\theta_{0}}>0$, the infimum is attained and
hence is a minimum. In contrast, when $\rho_{\theta_{0}}$ is not
faithful, the infimum may be finite without being attained; see Appendix
\ref{sec:nonattainment} for an explicit example.

Although the scalar optimal value
\begin{equation}
c_{\theta_{0}}:=\inf\left\{ \left.\Tr W\,V_{\theta_{0}}\left[M,\hat{\theta}\right]\right|\left(M,\hat{\theta}\right)\text{ is LUE},\,M\in\M(\H,s),\,s\in\N\right\} \label{eq:min}
\end{equation}
always exists, explicit formulas for $c_{\theta_{0}}$ and the corresponding
optimal locally unbiased estimators are known only for a limited class
of models, including pure-state models \cite{matsumoto_pure} and
two-level systems \cite{gill_massar,yama_tomo}. In principle, one
could therefore seek the optimal value in (\ref{eq:min}) by numerical
optimization. However, such an optimization cannot be carried out
directly: the number of outcomes $s$ of POVMs is not a priori bounded,
and the resulting search space is prohibitively large. The main objective
of this paper is to derive a sufficient upper bound on the support
size $s$ sufficient to preserve the optimal value  in (\ref{eq:min}).
Previously known bounds include $\frac{1}{2}d(d+1)\left(\dim\H\right)^{2}+1$,
obtained via conic programming \cite{conic}, and $\left(\dim\H\right)^{2}+d(d+1)$,
proved earlier by Fujiwara \cite{adaptive}. The present results are
obtained by extending Fujiwara\textquoteright s convex-analytic approach:
we apply the framework directly to the target information quantity
and, in addition, incorporate sufficient operator subspaces. This
yields the improved support-size bounds proved in this paper.

It is useful to clarify why the local multiparameter problem does
not reduce directly to finite-outcome arguments based on the convex
structure of POVMs, such as those developed in \cite{Chiribella}.
The classical Fisher information matrix is convex, in the Loewner
order, with respect to ordinary convex combinations of POVMs. In the
single-parameter case, minimizing the inverse Fisher information is
equivalent to maximizing the scalar Fisher information, and an optimal
measurement can therefore be chosen among extreme POVMs. In the multiparameter
case, however, the relevant objective is the weighted trace of the
inverse Fisher information matrix. This minimization problem cannot,
in general, be reduced to the maximization of a convex scalar function
of the POVM. Consequently, an optimal measurement need not be an extreme
POVM, and randomization among measurements may strictly decrease the
objective value. Indeed, for estimation problems on two-level systems,
the minimum can be attained by randomizing among projective measurements,
with the mixing probabilities depending on the weight matrix $W$
\cite{gill_massar,yama_tomo}. 

To state the main result, we introduce the infimum of the weighted
trace of the MSE under the restriction that the POVM belongs to $\M\left(\B,s\right)$:
\[
c_{\theta_{0}}\left(\B,s\right):=\inf\left\{ \left.\Tr WV_{\theta_{0}}\left[M,\hat{\theta}\right]\right|\left(M,\hat{\theta}\right)\text{ is LUE},\,M\in\M\left(\B,s\right)\right\} ,
\]
where $\M\left(\B,s\right)$ denotes the set of POVMs with outcome
set $\Omega_{s}$ such that each POVM element belongs to a subset
$\B\subset\B(\H)$.  Then $c_{\theta_{0}}=\inf_{s\in\N}c_{\theta_{0}}\left(\B(\H),s\right)$.
Next, we define a sufficient subspace $\A\subset\B(\H)$, as follows:
\begin{llllist}{00.00.0000}
\item [{(i)}] $\A$ is a real linear subspace satisfying $I\in\A$.
\item [{(ii)}] There exists a positive map $\Gamma:\B(\H)\to\A$ such that
$\Gamma(I)=I$ and
\begin{equation}
{\rm Re}\,\tr\rho_{\theta}\Gamma(B)={\rm Re}\,\tr\rho_{\theta}B\qquad\forall\theta\in\Theta,\forall B\in\B(\H).\label{eq:sufficient_space}
\end{equation}
\end{llllist}
By replacing condition (\ref{eq:sufficient_space}) with the following
conditions, we define a locally sufficient subspace $\A\subset\B(\H)$
:
\begin{align}
{\rm Re}\,\tr\rho_{\theta_{0}}\Gamma(B) & ={\rm Re}\,\tr\rho_{\theta_{0}}B\qquad\forall B\in\B(\H),\label{eq:lsufficient_space1}\\
{\rm Re}\,\tr\partial_{i}\rho_{\theta_{0}}\Gamma(B) & ={\rm Re}\,\tr\partial_{i}\rho_{\theta_{0}}B\qquad\forall B\in\B(\H),\text{1\ensuremath{\leq}\ensuremath{i}}\leq d.\label{eq:lsufficient_space2}
\end{align}
Here, by a positive map $\Gamma:\B(\H)\to\A$, we mean an $\R$-linear
map such that $\Gamma(B)\ge0$ for every $B\in\B(\H)$ with $B\ge0$.
This is weaker than requiring a completely positive map on $\B(\H)$,
but it is sufficient for our purposes. Note that this notion of a
sufficient subspace generalizes the notion of a sufficient subalgebra
defined via completely positive maps, introduced by Petz \cite{sufficient}.
For $\A\subset\B(\H)$, define
\begin{align*}
\A_{h} & :=\left\{ A\in\A\mid A=A^{*}\right\} ,\\
\A_{+} & :=\left\{ A\in\A\mid A\ge0\right\} ,\\
\A_{1} & :=\left\{ A\in\A_{+}\mid\tr A\leq1\right\} ,\\
\A_{e} & :=\left\{ A\in\A_{+}\mid\exists r>0\text{ such that }r\,A\text{ is an extreme point of }\A_{1}\right\} .
\end{align*}

We are now ready to state the main results. Suppose that there exists
a locally sufficient subspace $\A\subset\B(\H)$. Then we obtain a
sufficient support size for the local optimization problem and show
that
\begin{equation}
c_{\theta_{0}}=c_{\theta_{0}}\left(\A_{e},\dim\A_{h}+\frac{1}{2}d(d+1)-1\right).\label{eq:main_point}
\end{equation}
Note that when $\A=\B(\H)$, we have $\dim\A_{h}=(\dim\H)^{2}$. If
$\A$ is a real $*$-subalgebra, $\dim\A_{h}$ can be expressed more
explicitly; see Appendix \ref{sec:real_algebra} for details. More
generally, by specializing $\A$ to an appropriate real Jordan subalgebra
compatible with the model, one can often further reduce $\dim\A_{h}$
and obtain sharper support-size bounds. For details on real Jordan
algebras and their role in quantum sufficiency, see \cite{jordan-yamagata}.
This result has two important implications. First, the sufficient
support size can be smaller than previously known bounds. Second,
the optimization can be restricted to POVMs in $\A_{e}$. Since $\A_{e}$
can be parameterized with far fewer variables than $\A_{+}$ and $\B(\H)_{+}$,
this suggests the possibility of more efficient algorithms for finding
optimal POVMs.

In Bayesian estimation, the full-space finite-outcome reduction is
in line with general convex-analytic arguments for POVMs \cite{Chiribella}.
The present paper refines this reduction by incorporating sufficient
subspaces: when a sufficient subspace $\A$ is available, the optimization
can be restricted to $\A_{e}$, and the support size is bounded by
$\dim\A_{h}$.

Let $\left\{ \rho_{\theta}\in\S(\H)\mid\theta\in\Theta\subset\R^{d}\right\} $
be a parametric family of density operators on $\H$. Let $\pi$ be
a given prior on $\Theta$, and let $\ell_{\theta}:\R^{d}\to[0,\infty)$
be a loss function for each $\theta\in\Theta$. The average cost of
an estimator $\left(M,\hat{\theta}\right)$ is defined by
\[
c_{\pi}\left[M,\hat{\theta}\right]=\int_{\Theta}d\theta\,\pi(\theta)\,\sum_{x}\ell_{\theta}(\hat{\theta}(x))\,\tr\rho_{\theta}M_{x}.
\]
In this case, the map $\theta\mapsto\rho_{\theta}$ need not be smooth,
and no unbiasedness condition is imposed on $\left(M,\hat{\theta}\right)$.
To compute
\[
c_{\pi}:=\inf\left\{ \left.c_{\pi}\left[M,\hat{\theta}\right]\right|\left(M,\hat{\theta}\right)\text{ is an estimator},\,M\in\M(\H,s),\,s\in\N\right\} ,
\]
we seek a sufficient bound on the support size of POVM. For a subset
$\B\subset\B(\H)$, define the restricted optimal value by
\[
c_{\pi}(\B,s):=\inf\left\{ \left.c_{\pi}\left[M,\hat{\theta}\right]\right|\left(M,\hat{\theta}\right)\text{ is an estimator},\,M\in\M(\B,s)\right\} .
\]
Then $c_{\pi}=\inf_{s\in\N}c_{\pi}(\B(\H),s)$. If there exists a
sufficient subspace $\A\subset\B(\H)$, we can derive a sufficient
support-size bound for the value of the Bayesian optimization problem,
and we show that
\begin{equation}
c_{\pi}=c_{\pi}\left(\A_{e},\dim\A_{h}\right).\label{eq:main_prior}
\end{equation}
Since, for a general loss function, an optimal estimate for each measurement
outcome need not exist, we formulate the Bayesian result in terms
of the infimum of the Bayes risk.

This paper is organized as follows. In Section \ref{sec:point},
we study support-size bounds for POVMs in local estimation and prove
(\ref{eq:main_point}). In particular, we establish an affine-structure
argument for POVMs and classical Fisher information matrices, and
we prove uniqueness of the optimal classical Fisher information matrix
under weighted-trace minimization whenever the infimum is attained.
In Section \ref{sec:bayes}, we analyze the Bayesian setting and prove
(\ref{eq:main_prior}). In Section \ref{sec:Example}, we present
concrete examples, including the real two-level model and its two-copy
extension, to illustrate the resulting support-size bounds. Appendix
\ref{sec:real_algebra} describes the structure of finite-dimensional
real $*$-subalgebras and gives the corresponding dimension formula
for their Hermitian parts. Appendix \ref{sec:invertiblePOVM} gives
a local identifiability condition ensuring the existence of a POVM
with nonsingular classical Fisher information. Appendix \ref{sec:nonattainment}
gives an explicit example in which the infimum of the local optimization
problem is finite but is not attained.

\section{Sufficient support size of POVMs for LUE\label{sec:point}}

In this section, we study the support size of POVMs for locally unbiased
estimators and prove (\ref{eq:main_point}). Let $\H$ be a finite-dimensional
Hilbert space,  and let $\left\{ \rho_{\theta}\in\S(\H)\mid\theta\in\Theta\subset\R^{d}\right\} $
be a smooth parametric family of density operators on $\H$. For a
fixed POVM $M\in\M\left(\H,s\right)$, any locally unbiased estimator
$\left(M,\hat{\theta}\right)$ at $\theta_{0}\in\Theta$ \textcolor{red}{satisfies}
the classical Cram\'{e}r Rao inequality
\begin{equation}
V_{\theta_{0}}\left[M,\hat{\theta}\right]\geq F_{\theta_{0}}[M]^{-1},\label{eq:cramer_cl}
\end{equation}
where $F_{\theta_{0}}[M]$ is the classical Fisher information matrix,
\begin{equation}
F_{\theta_{0}}[M]=\sum_{x}g_{\theta_{0}}[M_{x}],\label{eq:cFisher}
\end{equation}
and
\begin{align*}
g_{\theta_{0}}[X]_{ij} & =\frac{\left(\tr\partial_{i}\rho_{\theta_{0}}X\right)\left(\tr\partial_{j}\rho_{\theta_{0}}X\right)}{\tr\rho_{\theta_{0}}X},\qquad X\in\B(\H)_{+}.
\end{align*}
Note that when $\tr\rho_{\theta_{0}}X=0$, then necessarily $\tr\partial_{i}\rho_{\theta_{0}}X=0$
for $i=1,\dots,d$. In this case, we define $g_{\theta_{0}}[X]=0$.
Whenever $F_{\theta_{0}}[M]$ is nonsingular, equality in (\ref{eq:cramer_cl})
is achievable by
\[
\hat{\theta}^{i}(x)=\theta^{i}_{0}+\sum^{d}_{j=1}\left(F_{\theta_{0}}[M]^{-1}\right)^{ij}\frac{\tr\partial_{j}\rho_{\theta_{0}}M_{x}}{\tr\rho_{\theta_{0}}M_{x}}\qquad(i=1,\dots,d),
\]
which satisfies the locally unbiased conditions (\ref{eq:lue1}) and
(\ref{eq:lue2}). Therefore, minimizing $\Tr W\,V_{\theta_{0}}\left[M,\hat{\theta}\right]$
reduces to minimizing $\Tr W\,F_{\theta_{0}}[M]^{-1}$. To analyze
this minimization, we first introduce the following lemmas.
\begin{lem}
\label{lem:g_property}For $X,Y\in\B_{+}(\H)$ and $t\in[0,\infty)$,
the following hold:
\end{lem}

\begin{description}
\item [{(i)}] $g_{\theta_{0}}[tX]=tg_{\theta_{0}}[X]$.
\item [{(ii)}] $g_{\theta_{0}}[X+Y]\leq g_{\theta_{0}}[X]+g_{\theta_{0}}[Y]$.
\end{description}
\begin{proof}
Assertion (i) is immediate from the definition. 

To prove (ii), it suffices to show that, for every $v\in\R^{d}$,
\[
v^{\mathsf{T}}\left\{ g_{\theta_{0}}[X]+g_{\theta_{0}}[Y]-g_{\theta_{0}}[X+Y]\right\} v\geq0.
\]
Set
\[
p[X]=\tr\rho_{\theta_{0}}X,\qquad D_{v}[X]=\sum^{d}_{i=1}v^{i}\tr\partial_{i}\rho_{\theta_{0}}X.
\]
If $p[X]=0$ or $p[Y]=0$, assertion (ii) follows immediately. We
may therefore assume that $p[X]>0$ and $p[Y]>0$. Then
\[
v^{\mathsf{T}}\left\{ g_{\theta_{0}}[X]+g_{\theta_{0}}[Y]-g_{\theta_{0}}[X+Y]\right\} v=\frac{\left(D_{v}[Y]p[X]-D_{v}[X]p[Y]\right)^{2}}{p[X]p[Y](p[X]+p[Y])}\geq0.
\]
This proves (ii).
\end{proof}

\begin{lem}
\label{lem:F_property}For a POVM $M=\left(M_{1},M_{2},\dots,M_{s}\right)\in\M(\H,s)$,
the following hold.
\end{lem}

\begin{description}
\item [{(i)}] For any permutation $\sigma\in S_{s}$, let $M^{\sigma}:=\left(M_{\sigma(1)},M_{\sigma(2)},\dots,M_{\sigma(s)}\right)\in\M(\H,s)$.
Then $F_{\theta_{0}}[M]=F_{\theta_{0}}[M^{\sigma}]$.
\item [{(ii)}] Let $N=\left(M_{1},M_{2},\dots,M_{s-2},M_{s-1}+M_{s}\right)\in\M(\H,s-1)$.
Then $F_{\theta_{0}}[M]\geq F_{\theta_{0}}[N].$
\end{description}
\begin{proof}
Part (i) is immediate from the definition (\ref{eq:cFisher}), which
is invariant under reordering. Part (ii) follows immediately from
Lemma \ref{lem:g_property} (ii). 
\end{proof}

This lemma can also be proved directly using the monotonicity of Fisher
metrics. By this lemma, the optimization can be restricted to rank-one
POVMs.

We are now ready to prove the main statement (\ref{eq:main_point}).
Suppose there exists a locally sufficient subspace $\A\subset\B(\H)$.
We establish the following theorem. 
\begin{thm}
\label{thm:mainLUE}For any $M\in\M(\H,s)$, there exists $N\in\M(\A_{+},\dim\A_{h}+d(d+1)/2)$
such that
\begin{equation}
F_{\theta_{0}}[N]=F_{\theta_{0}}[M].\label{eq:equiN}
\end{equation}
 Furthermore, there exists $N'\in\M(\A_{e},\dim\A_{h}+d(d+1)/2-1)$
such that
\begin{equation}
F_{\theta_{0}}[N']\geq F_{\theta_{0}}[M].\label{eq:upperN}
\end{equation}
\end{thm}

\begin{proof}
We first prove the existence of a POVM $N\in\M(\A_{+},\dim\A_{h}+d(d+1)/2)$
satisfying (\ref{eq:equiN}). Let $\Gamma:\B(\H)\to\A$ be a positive
map given at (\ref{eq:lsufficient_space1}) and (\ref{eq:lsufficient_space2}).
Since $\Gamma(M_{x})\in\A_{+}$ for each $x\in\Omega_{s}$, we have
$M':=\left\{ \Gamma(M_{x})\right\} ^{s}_{x=1}\in\M(\A_{+},s)$. By
(\ref{eq:lsufficient_space1}) and (\ref{eq:lsufficient_space2}),
it follows that $F_{\theta_{0}}[M']=F_{\theta_{0}}[M]$. For each
POVM element $\Gamma(M_{x})$, the pair $\left(\Gamma(M_{x}),g_{\theta_{0}}[\Gamma(M_{x})]\right)$
belongs to the real linear space $\A_{h}\oplus\B^{d}_{hr}$, where
$\B^{d}_{hr}$ is the space of $d\times d$ real symmetric matrices.
Hence $\dim\left(\A_{h}\oplus\B^{d}_{hr}\right)=\dim\A_{h}+d(d+1)/2$.
Therefore, if $s>\dim\A_{h}+d(d+1)/2$, the set $\left\{ \left(\Gamma(M_{x}),g_{\theta_{0}}[\Gamma(M_{x})]\right)\right\} ^{s}_{x=1}$
is linearly dependent. So there exist $\alpha_{x}\in\R$ ( $1\leq x\leq s$),
not all zero, such that
\begin{align}
\sum^{s}_{x=1}\alpha_{x}\Gamma(M_{x}) & =0,\label{eq:depend1}\\
\sum^{s}_{x=1}\alpha_{x}g_{\theta_{0}}\left[\Gamma(M_{x})\right] & =0.\label{eq:depend2}
\end{align}
Define
\[
M^{(t)}_{x}=(1-t\,\alpha_{x})\Gamma(M_{x}),
\]
and choose $x_{*}=\argmax_{x}\alpha_{x}$ with $t=1/\alpha_{x_{*}}$.
From (\ref{eq:depend1}),
\[
\sum^{s}_{x=1}M^{(t)}_{x}=\sum^{s}_{x=1}(1-t\,\alpha_{x})\Gamma(M_{x})=I.
\]
Also, $M^{(t)}_{x_{*}}=0$ and $M^{(t)}_{x}\geq0$ for all $x$. By
Lemma \ref{lem:g_property} (i) and (\ref{eq:depend2}), we have
\begin{align*}
F_{\theta_{0}}[M^{(t)}] & =\sum^{s}_{x=1}g_{\theta_{0}}\left[M^{(t)}_{x}\right]=\sum^{s}_{x=1}g_{\theta_{0}}\left[(1-t\,\alpha_{x})\Gamma(M_{x})\right]\\
 & =\sum^{s}_{x=1}(1-t\,\alpha_{x})g_{\theta_{0}}\left[\Gamma(M_{x})\right]=F_{\theta_{0}}[M].
\end{align*}
Removing the zero element $M^{(t)}_{x_{*}}$, we obtain a POVM belonging
to $\M(\A_{+},s-1)$ with the same Fisher matrix. Repeating this reduction
yields $N\in\M(\A_{+},\dim\A_{h}+d(d+1)/2)$ such that (\ref{eq:equiN}). 

We next prove the existence of a POVM $N'\in\M(\A_{e},\dim\A_{h}+d(d+1)/2-1)$
satisfying (\ref{eq:upperN}). Using Lemma \ref{lem:g_property} (ii)
together with an extreme-point decomposition of each POVM element
of $N$, we obtain a POVM $M''\in\M(\A_{e},s'')$ such that $F_{\theta_{0}}[M'']\geq F_{\theta_{0}}[N]=F_{\theta_{0}}[M]$.
Let $g_{\theta_{0}}[M''_{x}]_{-}$ denote the collection obtained
from $g_{\theta_{0}}[M''_{x}]$ by removing its (1,1)-entry. Then
the pair $\left(M''_{x},g_{\theta_{0}}[M''_{x}]_{-}\right)$ belongs
to a real linear space of dimension $\dim\A_{h}+d(d+1)/2-1=\dim\A_{h}+d(d+1)/2-1$.
Thus, if $s''>\dim\A_{h}+d(d+1)/2-1$, there exist $\alpha_{x}\in\R$
( $1\leq x\leq s''$) and $r\in\R_{\geq0}$ such that
\begin{align}
\sum^{s''}_{x=1}\alpha_{x}M''_{x} & =0,\label{eq:depend=00FF13}\\
\sum^{s''}_{x=1}\alpha_{x}g_{\theta_{0}}\left[M''_{x}\right] & =-rE_{11},\label{eq:depend4}
\end{align}
where $E_{11}$ is the matrix with $1$ in the (1,1)-entry and $0$
elsewhere. Set $M^{(t)}_{x}=(1-t\,\alpha_{x})M''_{x}$, choose $x_{*}=\argmax_{x}\alpha_{x}$,
and put $t=1/\alpha_{x_{*}}$. From (\ref{eq:depend=00FF13}), $\sum^{s''}_{x=1}M^{(t)}_{x}=\sum^{s''}_{x=1}(1-t\,\alpha_{x})M''_{x}=I$,
and $M^{(t)}_{x_{*}}=0$, and $M^{(t)}_{x}\geq0$ for all $x$. By
(\ref{eq:depend4}),
\begin{align*}
F_{\theta_{0}}[M^{(t)}] & =\sum^{s''}_{x=1}(1-t\,\alpha_{x})g_{\theta_{0}}\left[M''_{x}\right]\\
 & =F_{\theta_{0}}[M'']+t\,rE_{11}\geq F_{\theta_{0}}[M''].
\end{align*}
Removing the zero element and repeating the reduction, we obtain $N'\in\M(\A_{e},\dim\A_{h}+d(d+1)/2-1)$
that satisfies (\ref{eq:upperN}). This completes the proof. 
\end{proof}

By Theorem \ref{thm:mainLUE}, equation (\ref{eq:main_point}) follows
immediately.
\begin{cor}
\label{cor:uniqueF}Let $c_{\theta_{0}}=\inf\left\{ \Tr WF_{\theta_{0}}\left[M\right]^{-1}\mid M\in\M\left(\H,s\right),s\in\N\right\} $.
If two POVMs $M$ and $M'$ satisfy $c_{\theta_{0}}=\Tr WF_{\theta_{0}}\left[M\right]^{-1}=\Tr WF_{\theta_{0}}\left[M'\right]^{-1}$,
then $F_{\theta_{0}}\left[M\right]=F_{\theta_{0}}\left[M'\right]$.
Moreover, there exists a POVM\\
$N\in\M\left(\A_{e},\dim\A_{h}+d(d+1)/2-1\right)$ such that $c_{\theta_{0}}=\Tr WF_{\theta_{0}}\left[N\right]^{-1}$. 
\end{cor}

\begin{proof}
Since the latter statement follows immediately from Theorem \ref{thm:mainLUE},
we prove only the uniqueness claim $F_{\theta_{0}}\left[M\right]=F_{\theta_{0}}\left[M'\right]$.
Assume $\Tr WF_{\theta_{0}}\left[M\right]^{-1}=\Tr WF_{\theta_{0}}\left[M'\right]^{-1}=c_{\theta_{0}}$
but $F_{\theta_{0}}\left[M\right]\not=F_{\theta_{0}}\left[M'\right]$.
We derive a contradiction. Define
\[
\frac{1}{2}\left(M\oplus M'\right)=\frac{1}{2}\left(M_{1},M_{2},\dots,M'_{1},M'_{2},\dots\right).
\]
From Lemma \ref{lem:g_property} (i) , we have
\[
F_{\theta_{0}}\left[\frac{1}{2}\left(M\oplus M'\right)\right]=\frac{1}{2}\left(F_{\theta_{0}}\left[M\right]+F_{\theta_{0}}\left[M'\right]\right).
\]
Since $t\mapsto1/t$ is a strictly operator convex function,
\begin{align*}
\Tr WF_{\theta_{0}}\left[\frac{1}{2}\left(M\oplus M'\right)\right]^{-1} & =\Tr W\left\{ \frac{1}{2}\left(F_{\theta_{0}}\left[M\right]+F_{\theta_{0}}\left[M'\right]\right)\right\} ^{-1}\\
 & <\Tr W\left\{ \frac{1}{2}\left(F_{\theta_{0}}\left[M\right]^{-1}+F_{\theta_{0}}\left[M'\right]^{-1}\right)\right\} .
\end{align*}
Therefore,
\[
\Tr WF_{\theta_{0}}\left[\frac{1}{2}\left(M\oplus M'\right)\right]^{-1}<c_{\theta_{0}}
\]
which contradicts the definition of $c_{\theta_{0}}$.
\end{proof}

\section{Sufficient support size of POVM for Bayes estimation\label{sec:bayes}}

In this section, we study the support size of POVMs for Bayesian estimators
and prove (\ref{eq:main_prior}). In Bayesian estimation, finite-outcome
reductions are naturally connected with convex-analytic properties
of POVMs, as illustrated for general loss functions in \cite{Chiribella}.
The present result refines this viewpoint by incorporating sufficient
subspaces: when a sufficient subspace $\A$ is available, the optimization
can be restricted to $\A_{e}$, and the required support size is bounded
by $\dim\A_{h}$.

Let $\left\{ \rho_{\theta}\in\S(\H)\mid\theta\in\Theta\subset\R^{d}\right\} $
be a parametric family of density operators on a finite-dimensional
Hilbert space $\H$. The mapping $\theta\mapsto\rho_{\theta}$ is
not necessarily smooth, and no unbiasedness condition is imposed on
an estimator $\left(M,\hat{\theta}\right)$. Let $\pi$ be a given
prior on $\Theta$, and let $\ell_{\theta}:\R^{d}\to[0,\infty)$ be
a loss function for each $\theta\in\Theta$. To study the optimal
value of the average cost
\begin{align*}
c_{\pi}\left[M,\hat{\theta}\right] & =\int_{\Theta}d\theta\,\pi(\theta)\,\mathop{\normalcolor \sum^{s}_{x=1}}{\normalcolor \ell_{\theta}(\hat{\theta}(x))\,\tr\rho_{\theta}M_{x}}\\
 & =\sum^{s}_{x=1}h\left[M_{x},\hat{\theta}(x)\right],
\end{align*}
we derive a sufficient support size for POVMs, where
\[
h\left[X,v\right]=\int_{\Theta}d\theta\,\pi(\theta)\,\ell_{\theta}(v)\,\tr\rho_{\theta}X\qquad(X\in\B_{+}(\H),v\in\R^{d}).
\]
It is immediate that, the map
\[
X\to h[X]=\inf_{v\in\R^{d}}h\left[X,v\right]
\]
satisfies the following lemma. 
\begin{lem}
\label{lem:gb_property}$h[tX]=t\,h[X]$ for all $X\in\B_{+}(\H)$,
$t\geq0$. 
\end{lem}

Since the estimator values $\hat{\theta}(x)$ can be optimized independently
for different outcomes, we have
\[
c_{\pi}[M]:=\inf_{\hat{\theta}}c_{\pi}[M,\hat{\theta}]=\sum^{s}_{x=1}h[M_{x}]
\]
for a POVM $M\in\M(\H,s)$. Then following holds. 
\begin{lem}
\label{lem:dev_b}For a POVM $M=\left(M_{1},M_{2},\dots,M_{s}\right)\in\M(\H,s)$,
the following hold:
\end{lem}

\begin{description}
\item [{(i)}] For any permutation $\sigma\in S_{s}$, let $M^{\sigma}:=\left(M_{\sigma(1)},M_{\sigma(2)},\dots,M_{\sigma(s)}\right)\in\M(\H,s)$.
Then $c_{\pi}[M]=c_{\pi}[M^{\sigma}]$.
\item [{(ii)}] Let $N=\left(M_{1},M_{2},\dots,M_{s-1}+M_{s}\right)\in\M(\H,s-1)$.
Then $c_{\pi}[M]\leq c_{\pi}[N]$.
\end{description}
\begin{proof}
To prove (i), observe that permuting outcome labels does not change
the optimization problem that defines $c_{\pi}[M]$, and hence $c_{\pi}[M]=c_{\pi}[M^{\sigma}]$.
To prove (ii), let $\epsilon>0$. By the definition of $c_{\pi}[N]$,
we can choose $\hat{\theta}$ such that $c_{\pi}\left[N,\hat{\theta}\right]<c_{\pi}\left[N\right]+\epsilon$.
Define $\hat{\theta}':\Omega_{s}\to\R^{d}$ by
\[
\hat{\theta}'(x)=\begin{cases}
\hat{\theta}(x), & 1\leq x\leq s-2,\\
\hat{\theta}(s-1), & x=s-1\text{ or }x=s.
\end{cases}
\]
Then the estimator $\left(M,\hat{\theta}'\right)$ induces exactly
the same average cost as $\left(N,\hat{\theta}\right)$, so $c_{\pi}[M]\leq c_{\pi}\left[M,\hat{\theta}'\right]=c_{\pi}\left[N,\hat{\theta}\right]<c_{\pi}\left[N\right]+\epsilon$.
Since $\epsilon>0$ is arbitrary, we obtain $c_{\pi}[M]\leq c_{\pi}[N]$.
\end{proof}

We are now ready to prove (\ref{eq:main_prior}). Suppose there exists
a sufficient subspace $\A\subset\B(\H)$. We establish the following
theorem.
\begin{thm}
\label{thm:mainB}For any $M\in\M(\H,s)$, there exists $N\in\M(\A_{e},\dim\A_{h})$
such that
\begin{equation}
c_{\pi}[N]\leq c_{\pi}[M].\label{eq:upperNb}
\end{equation}
\end{thm}

\begin{proof}
Let $\Gamma:\B(\H)\to\A$ be a positive map given at (\ref{eq:sufficient_space}).
Since $\Gamma(M_{x})\in\A_{+}$, for each $x\in\Omega_{s}$, we have
$M':=\left\{ \Gamma(M_{x})\right\} ^{s}_{x=1}\in\M(\A_{+},s)$. By
(\ref{eq:sufficient_space}), it follows that $c_{\pi}\left[M'\right]=c_{\pi}\left[M\right]$.
Using Lemma \ref{lem:dev_b} (ii) together with an extreme-point decomposition
of each POVM element of $M'$, we obtain a POVM $M''\in\M(\A_{e},s'')$
such that $c_{\pi}[M'']\leq c_{\pi}[M']=c_{\pi}[M]$. Since $\A_{h}$
is a $\dim\A_{h}$-dimensional real linear space, the set $\left\{ \left(M''_{x}\right)\right\} ^{s''}_{x=1}$
is linearly dependent whenever $s''>\dim\A_{h}$. Thus, when $s''>\dim\A_{h}$,
there exist $\alpha_{x}\in\R$ ( $1\leq x\leq s''$) and $r\in\R_{\geq0}$
such that
\begin{align}
\sum^{s''}_{x=1}\alpha_{x}M''_{x} & =0,\label{eq:depend_b1}\\
\sum^{s''}_{x=1}\alpha_{x}h\left[M''_{x}\right] & =r.\label{eq:depend_b2}
\end{align}
Let $M^{(t)}_{x}=(1-t\,\alpha_{x})M''_{x}$, and define $x_{*}=\argmax_{x}\alpha_{x}$
and $t=1/\alpha_{x_{*}}$. From (\ref{eq:depend_b1}), $\sum^{s''}_{x=1}M^{(t)}_{x}=\sum^{s''}_{x=1}(1-t\,\alpha_{x})M''_{x}=I$
and moreover $M^{(t)}_{x_{*}}=0$ and $M^{(t)}_{x}\geq0$ for all
$x$. From Lemma \ref{lem:gb_property} and (\ref{eq:depend_b2}),
we have
\begin{align*}
c_{\pi}[M^{(t)}] & =\sum^{s''}_{x=1}h\left[(1-t\,\alpha_{x})M''_{x}\right]\\
 & =\sum^{s''}_{x=1}(1-t\,\alpha_{x})h\left[M''_{x}\right]\\
 & =c_{\pi}[M'']-t\,r\leq c_{\pi}[M''].
\end{align*}
Since $M^{(t)}_{x_{*}}=0$, removing this element gives a POVM in
$\M(\A_{e},s''-1)$. Repeating this procedure, we obtain a POVM $N\in\M(\A_{e},\dim\A_{h})$
that satisfies (\ref{eq:upperNb}). 
\end{proof}

By Theorem \ref{thm:mainB}, equation (\ref{eq:main_prior}) follows
immediately.

\section{Example: Real two-level model\label{sec:Example}}

In this section, we illustrate our bounds with a real two-level (qubit)
model and its i.i.d. two-copy extension. 

\subsection*{Single-copy model}

Consider the family on $\H=\C^{2}$:
\[
\left\{ \rho_{\theta}=\frac{1}{2}(I+xX+zZ)\mid\theta=(x,z)\in\Theta\right\} ,\qquad\Theta=\left\{ \theta\in\R^{2}\mid\left\Vert \theta\right\Vert <1\right\} 
\]
Since $\rho_{\theta}$ is a real matrix for each $\theta$, the space
of all real $2\times2$ matrices $\A=\mathbb{M}_{2}(\R)$ is a locally
sufficient subspace in the sense of (\ref{eq:lsufficient_space1})(\ref{eq:lsufficient_space2})
for any fixed $\theta_{0}\in\Theta$. It is also a real $*$-subalgebra
(see Appendix \ref{sec:real_algebra}). Hence $\dim\A_{h}=3$, and
Corollary \ref{cor:uniqueF} implies that, for any $2\times2$ weight
matrix $W$, an optimal measurement attaining $c_{\theta_{0}}=\min_{M}\Tr WF_{\theta_{0}}\left[M\right]^{-1}$
can be chosen from $\M\left(\A_{e},3+\frac{2(2+1)}{2}-1\right)=\M\left(\A_{e},5\right)$.
Since $\rho_{\theta_{0}}>0$ in this model, the minimum exists. In
this qubit two-parameter setting, the known result is sharper: the
value of $c_{\theta_{0}}$ is known analytically, and an optimal measurement
is achievable with support size 4 \cite{gill_massar,yama_tomo}. For
the Bayesian criterion, Theorem \ref{thm:mainB} yields that, for
any prior $\pi$ and loss functions $\ell_{\theta}$, the infimum
of $c_{\pi}[M]$ over all POVMs is equal to the infimum over POVMs
in $\M(\A_{e},3)$. This type of reduction for real-valued models
has also been used in the numerical optimization of POVMs in quantum
parameter estimation \cite{kimizu}. It is closely related to the
sufficient-subspace viewpoint used in the present paper.

\subsection*{Two-copy extension}

Next consider the two-copy model $\left\{ \rho^{\otimes2}_{\theta}\mid\theta=(x,z)\in\Theta\right\} $
on $\mathcal{H}^{\otimes2}$. In this case, an analytic expression
for $c_{\theta_{0}}$ is not known. Let $\{|0\rangle,|1\rangle\}$
be the standard basis of $\mathcal{H}$, and use the basis of $\mathcal{H}^{\otimes2}$:$\ket{00},\ket{11},\frac{1}{\sqrt{2}}\{\ket{01}+\ket{10}\},\frac{1}{\sqrt{2}}\{\ket{01}-\ket{10}\}$.
In this basis, the matrix representation of $\rho^{\otimes2}_{\theta}$
is
\[
\frac{1}{4}\begin{pmatrix}(1+z)^{2} & x^{2} & \sqrt{2}x(1+z) & 0\\
x^{2} & (1-z)^{2} & \sqrt{2}x(1-z) & 0\\
\sqrt{2}x(1+z) & \sqrt{2}x(1-z) & 1+x^{2}-z^{2} & 0\\
0 & 0 & 0 & 1-x^{2}-z^{2}
\end{pmatrix}.
\]
Therefore, $\A=\mathbb{M}_{3}(\R)\oplus\mathbb{M}_{1}(\R)$ is a locally
sufficient subspace, and $\dim\A_{h}=7$. Hence, by Corollary \ref{cor:uniqueF},
for any $2\times2$ weight matrix $W$, an optimizer of $c_{\theta_{0}}=\min_{M}\Tr WF_{\theta_{0}}\left[M\right]^{-1}$
can be chosen from $\M\left(\A_{e},7+\frac{2(2+1)}{2}-1\right)=\M\left(\A_{e},9\right)$.
Since $\rho^{\otimes2}_{\theta_{0}}>0$ in this model, the minimum
exists. Furthermore, by Theorem \ref{thm:mainB}, for any prior $\pi$
and loss functions $\ell_{\theta}$, the infimum of $c_{\pi}[M]$
over all POVMs is equal to the infimum over POVMs in $\M(\A_{e},7)$. 

\section{Conclusion}

We studied finite-outcome reduction for quantum estimation problems
by focusing on the support size of POVMs required to attain optimal
performance.

For local estimation, we considered minimization of the weighted trace
of the MSE under the locally unbiased condition. We showed that if
a locally sufficient subspace $\mathcal{A}\subset\mathcal{B}(\mathcal{H})$
exists, then the optimization can be reduced to POVMs supported on
$\mathcal{A}_{e}$, and the optimization can be restricted to POVMs
with at most $\dim\A_{h}+\frac{1}{2}d(d+1)-1$ outcomes. In particular,
when $\mathcal{A}=\mathcal{B}(\mathcal{H})$, we have $\dim\mathcal{A}_{h}=(\dim\mathcal{H})^{2}$,
so the bound becomes $(\dim\mathcal{H})^{2}+\frac{1}{2}d(d+1)-1$,
which is the general full-space upper bound. We also proved that,
whenever the infimum is attained, the optimal classical Fisher information
matrix is unique under weighted-trace minimization, clarifying the
structure of optimal solutions.

For Bayesian estimation, where neither smoothness of $\theta\mapsto\rho_{\theta}$
nor unbiasedness is required, we established an analogous finite-support
reduction: if a sufficient subspace $\mathcal{A}\subset\mathcal{B}(\mathcal{H})$
exists, then $c_{\pi}=c_{\pi}(\mathcal{A}_{e},\dim\mathcal{A}_{h})$,
so the infimum of the Bayes risk over all POVMs is equal to the infimum
obtained by restricting the POVM to $\mathcal{M}(\mathcal{A}_{e},\dim\mathcal{A}_{h})$.
Thus, both in local and Bayesian settings, optimization over all POVMs
can be replaced by optimization over a finite-dimensional, rank-one/extremal
model with explicit support-size guarantees.

It is useful to compare the present bounds with the support-size scale
that appears in general results on extreme POVMs. In finite dimensions,
extreme POVMs have support size at most $(\dim\mathcal{H})^{2}$,
and this often gives a finite-outcome reduction when the optimum of
the problem under consideration can be attained at an extreme POVM
\cite{Chiribella}. As discussed in the Introduction, however, the
local multiparameter objective considered here is not of this type:
the minimization of the weighted trace of the inverse classical Fisher
information matrix need not be attained at an extreme POVM, and randomization
among measurements may improve the objective value. Thus the usual
extremal-POVM argument alone does not settle the optimal support-size
question for the present local multiparameter problem. 

While our bounds on sufficient POVM support size for quantum estimation
are tighter than previously known ones, they should be understood
as sufficient upper bounds rather than tight characterizations. The
precise minimal support size required for optimal local multiparameter
estimation remains open. In particular, it is not known in general
whether the additional term $\frac{1}{2}d(d+1)-1$ in the full-space
bound $(\dim\mathcal{H})^{2}+\frac{1}{2}d(d+1)-1$ is necessary, nor
whether smaller model-dependent bounds can always be obtained. Similar
questions remain for the sufficient-subspace bound involving $\dim\A_{h}$.
Clarifying the tightness of these bounds, or finding further reductions
under additional structural assumptions, would be an interesting direction
for future work. Even so, establishing a principled support-size guarantee
is crucial: numerical optimization with an arbitrarily fixed number
of outcomes, without theoretical justification, cannot certify optimality.

\section*{ACKNOWLEDGMENTS}

This work was supported by JSPS KAKENHI Grant Numbers JP23H01090 and
JP22K03466 and JST ERATO Grant Number JPMJER2402, Japan.

\appendix

\section{Dimension of real $*$-subalgebra\label{sec:real_algebra}}

Let $\mathcal{A}$ be a subspace of a finite-dimensional matrix algebra
$\mathcal{\B}(\mathcal{H})$. We call $\mathcal{A}$ a real {*}-subalgebra
if the following conditions hold:
\begin{description}
\item [{(i)}] If $A,B\in\A$, then $A+B\in\A$. 
\item [{(ii)}] If $A,B\in\A$, then $AB\in\A$.
\item [{(iii)}] If $A\in\A$ and $r\in\R$, then $rA\in\A$. 
\item [{(iv)}] If $A\in\A$, then $A^{*}\in\A$.
\item [{(v)}] $I\in\A$.
\end{description}
Note that (iii) requires closure only under real scalars ($r\in\mathbb{R}$),
not necessarily under complex scalars. 

It is known that, as a real {*}-algebra, $\mathcal{A}$
admits the block decomposition
\[
\A\cong\left\{ \bigoplus_{i}\mathbb{M}_{n_{i}}(\R)\right\} \oplus\left\{ \bigoplus_{j}\mathbb{M}_{n_{j}}(\C)\right\} \oplus\left\{ \bigoplus_{k}\mathbb{M}_{n_{k}}(\mathbb{H})\right\} .
\]
Here, $\mathbb{M}_{n}(\mathbb{R})$ is the algebra of $n\times n$
real matrices, $\mathbb{M}_{n}(\mathbb{C})$ is the algebra of $n\times n$
complex matrices, and $\mathbb{M}_{n}(\mathbb{H})$ is the algebra
of $n\times n$ quaternionic matrices. In the standard complex realization,
each quaternion entry is represented by a $2\times2$ complex matrix
of the form
\[
wI_{2}+i(xX+yY+zZ),\qquad w,x,y,z\in\mathbb{R},
\]
where $X,Y,Z$ are the Pauli matrices. 
Let $\mathcal{A}_{h}:=\{A\in\mathcal{A}\mid A^{*}=A\}$
be the Hermitian part of $\mathcal{A}$. Then
\[
\dim\A_{h}=\sum_{i}\frac{n_{i}(n_{i}+1)}{2}+\sum_{j}n^{2}_{j}+\sum_{k}\{2n^{2}_{k}-n_{k}\}.
\]

Although we have focused on real $*$-subalgebras in this appendix,
a finer analysis may be possible by considering real Jordan algebras.
In the irreducible decomposition of a finite-dimensional real Jordan
algebra, in addition to the Hermitian parts of real, complex, and
quaternionic matrix algebras, spin factors also appear. This may allow
further model-dependent reductions of the dimension entering the support-size
bounds. For details on real Jordan algebras and their role in quantum
sufficiency, see \cite{jordan-yamagata}.

\section{Local identifiability and nonsingular classical Fisher information\label{sec:invertiblePOVM}}

In this appendix, we spell out a simple local regularity condition
under which there exists a POVM whose classical Fisher information
matrix is nonsingular. This also implies that the local optimization
problem (\ref{eq:min}) has a finite value. 

Let
\[
\left\{ \rho_{\theta}\in\mathcal{S}(\mathcal{H})\mid\theta\in\Theta\subset\mathbb{R}^{d}\right\} 
\]
be a $d$-parameter family of density operators, and let $\theta_{0}$
be an interior point of $\Theta$. Suppose that the partial derivatives
\[
D_{i}:=\left.\frac{\partial\rho_{\theta}}{\partial\theta^{i}}\right|_{\theta=\theta_{0}},\qquad i=1,\ldots,d,
\]
exist in $\mathcal{B}_{h}(\mathcal{H})$. 

We say that the model is locally identifiable at $\theta_{0}$ if
$D_{1},\ldots,D_{d}$ are linearly independent over $\mathbb{R}$.
This condition is necessary for the existence of a POVM with nonsingular
classical Fisher information. Indeed, if there exists a nonzero vector
$v=(v_{1},\ldots,v_{d})\in\mathbb{R}^{d}$ such that
\[
\sum^{d}_{i=1}v_{i}D_{i}=0,
\]
then, for every POVM $M$, the corresponding classical Fisher information
matrix is singular.

Conversely, suppose that the model is locally identifiable at $\theta_{0}$.
Choose a finite informationally complete POVM $M=(M_{x})_{x\in\Omega}$,
that is, a finite POVM such that the linear map
\[
\Phi:\mathcal{B}_{h}(\mathcal{H})\to\mathbb{R}^{|\Omega|},\qquad\Phi(A):=\bigl(\tr AM_{x}\bigr)_{x\in\Omega}
\]
 is injective. Such a POVM exists in finite dimension. Put
\[
p_{x}:=\tr\rho_{\theta_{0}}M_{x},\qquad J_{xi}:=\tr D_{i}M_{x}.
\]
If $p_{x}=0$, then $J_{xi}=0$ for all $i$. Indeed, the function
\[
\theta\mapsto p_{\theta}(x):=\tr\rho_{\theta}M_{x}
\]
is nonnegative and attains its minimum value $0$ at the interior
point $\theta_{0}$. Hence all its first derivatives at $\theta_{0}$
vanish. We therefore define the classical Fisher information matrix
by
\[
F_{\theta_{0}}[M]_{ij}=\sum_{x:\,p_{x}>0}\frac{J_{xi}J_{xj}}{p_{x}},
\]
 equivalently adopting the convention that the term with $p_{x}=0$
is zero. For any $v=(v_{1},\ldots,v_{d})\in\mathbb{R}^{d}$, we have
\[
v^{\mathsf{T}}F_{\theta_{0}}[M]v=\sum_{x:\,p_{x}>0}\frac{\left(\tr\left(\sum^{d}_{i=1}v_{i}D_{i}\right)M_{x}\right)^{2}}{p_{x}}.
\]
 If $v\neq0$, then $\sum_{i}v_{i}D_{i}\neq0$ by local identifiability.
Since $\Phi$ is injective, there exists $x\in\Omega$ such that
\[
\tr\left(\sum^{d}_{i=1}v_{i}D_{i}\right)M_{x}\neq0.
\]
For this $x$, necessarily $p_{x}>0$, because all terms with $p_{x}=0$
have zero first derivative. Hence 
\[
v^{\mathsf{T}}F_{\theta_{0}}[M]v>0
\]
 for every nonzero $v$, and $F_{\theta_{0}}[M]$ is positive definite.
Consequently, 
\[
\Tr WF_{\theta_{0}}[M]^{-1}<\infty
\]
for every positive semidefinite weight matrix $W$. In particular,
under local identifiability at $\theta_{0}$, the local optimization
problem (\ref{eq:min}) has a finite value.

\section{Nonattainment of the infimum in local estimation\label{sec:nonattainment}}

In the local setting, when $\rho_{\theta_{0}}>0$, compactness of
$\M(\H,s)$ together with the continuity of $M\mapsto F_{\theta_{0}}[M]$
implies, after the finite support-size reduction, that the infimum
in (\ref{eq:min}) is attained. In contrast, when $\rho_{\theta_{0}}$
is not faithful, the infimum in (\ref{eq:min}) need not be attained.
We demonstrate this by constructing a smooth constant-rank model for
which nonattainment occurs even with a strictly positive definite
weight matrix.

Let $\{|0\rangle,|1\rangle,|2\rangle\}$ be an orthonormal basis of
$\mathbb{C}^{3}$, and put
\[
|+\rangle=\frac{|0\rangle+|1\rangle}{\sqrt{2}},\qquad|-\rangle=\frac{|0\rangle-|1\rangle}{\sqrt{2}}.
\]
Consider
\[
\widetilde{\rho}_{\theta^{1}}=\frac{1+\theta^{1}}{2}|0\rangle\langle0|+\frac{1-\theta^{1}}{2}|1\rangle\langle1|,\qquad|\theta^{1}|<1,
\]
and define
\[
K=\frac{i}{\sqrt{2}}\left(|2\rangle\langle+|-|+\rangle\langle2|\right),\qquad U_{\theta^{2}}=e^{-i\theta^{2}K}.
\]
We consider the smooth two-parameter model
\[
\rho_{\theta^{1},\theta^{2}}=U_{\theta^{2}}\widetilde{\rho}_{\theta^{1}}U^{\dagger}_{\theta^{2}}.
\]
Its rank is equal to $2$ for every $|\theta^{1}|<1$. We take $\theta_{0}=(0,0)$.
Then
\[
\rho:=\rho_{\theta_{0}}=\frac{1}{2}\left(|0\rangle\langle0|+|1\rangle\langle1|\right),
\]
and
\[
D_{1}:=\left.\partial_{1}\rho_{\theta}\right|_{\theta=0}=\frac{1}{2}\left(|0\rangle\langle0|-|1\rangle\langle1|\right),
\]
while
\[
D_{2}:=\left.\partial_{2}\rho_{\theta}\right|_{\theta=0}=\frac{1}{2\sqrt{2}}\left(|+\rangle\langle2|+|2\rangle\langle+|\right).
\]
The corresponding SLDs may be chosen as
\[
L_{1}=|0\rangle\langle0|-|1\rangle\langle1|,
\]
and
\[
L_{2}=\sqrt{2}\left(|+\rangle\langle2|+|2\rangle\langle+|\right).
\]
Hence the SLD Fisher information matrix at $\theta_{0}$ is
\begin{equation}
J_{\theta_{0}}=\begin{pmatrix}1 & 0\\
0 & 1
\end{pmatrix},\label{eq:A1}
\end{equation}
and the SLD Cram\'er--Rao inequality (\ref{eq:SLDCRbound}) implies
\begin{equation}
V_{\theta_{0}}[M,\hat{\theta}]\geq J^{-1}_{\theta_{0}}\label{eq:A_SLDbound}
\end{equation}
for every locally unbiased estimator $(M,\hat{\theta})$.

We now show that there exists a family of locally unbiased estimators
$(M^{(\varepsilon)},\hat{\theta}^{(\varepsilon)})$, indexed by $0<\varepsilon<1$,
such that
\[
\lim_{\varepsilon\to0}V_{\theta_{0}}[M^{(\varepsilon)},\hat{\theta}^{(\varepsilon)}]=J^{-1}_{\theta_{0}},
\]
whereas no locally unbiased estimator attains equality in (\ref{eq:A_SLDbound}).
Consequently, for every strictly positive definite weight matrix $W>0$,
the infimum
\[
c_{\theta_{0}}:=\inf\left\{ \left.\Tr W\,V_{\theta_{0}}\left[M,\hat{\theta}\right]\right|\left(M,\hat{\theta}\right)\text{ is LUE},\,M\in\M(\H,s),\,s\in\N\right\} 
\]
is finite but is not attained by any locally unbiased estimator.

For $0<\varepsilon<1$, define the five-outcome POVM
\[
M^{(\varepsilon)}_{1}=(1-\varepsilon)|0\rangle\langle0|,\qquad M^{(\varepsilon)}_{2}=(1-\varepsilon)|1\rangle\langle1|,\qquad M^{(\varepsilon)}_{3}=\varepsilon|-\rangle\langle-|,
\]
and
\begin{equation}
M^{(\varepsilon)}_{\pm}=\frac{1}{2}\left(\sqrt{\varepsilon}|+\rangle\pm|2\rangle\right)\left(\sqrt{\varepsilon}\langle+|\pm\langle2|\right).\label{eq:A2}
\end{equation}
A direct calculation shows that
\[
M^{(\varepsilon)}_{1}+M^{(\varepsilon)}_{2}+M^{(\varepsilon)}_{3}+M^{(\varepsilon)}_{+}+M^{(\varepsilon)}_{-}=I.
\]
The classical Fisher information matrix corresponding to this POVM
is
\[
F_{\theta_{0}}[M^{(\varepsilon)}]=\begin{pmatrix}1-\varepsilon & 0\\[6pt]
0 & 1
\end{pmatrix}.
\]
Since this matrix is nonsingular, there exists a locally unbiased
estimator $(M^{(\varepsilon)},\hat{\theta}^{(\varepsilon)})$ satisfying
the classical Cram\'er--Rao equality,
\begin{equation}
V_{\theta_{0}}[M^{(\varepsilon)},\hat{\theta}^{(\varepsilon)}]=F_{\theta_{0}}[M^{(\varepsilon)}]^{-1}=\begin{pmatrix}\dfrac{1}{1-\varepsilon} & 0\\[6pt]
0 & 1
\end{pmatrix}.\label{eq:A4}
\end{equation}
Moreover,
\[
\lim_{\varepsilon\to0}F_{\theta_{0}}[M^{(\varepsilon)}]^{-1}=J^{-1}_{\theta_{0}}.
\]

It remains to show that no locally unbiased estimator attains equality
in (\ref{eq:A_SLDbound}). Let $M=\{M_{x}\}$ be an arbitrary finite-outcome
POVM. If equality were attained, then
\[
J^{-1}_{\theta_{0}}=V_{\theta_{0}}[M,\hat{\theta}]\geq F_{\theta_{0}}[M]^{-1}\geq J^{-1}_{\theta_{0}},
\]
where the second inequality follows from $F_{\theta_{0}}[M]\leq J_{\theta_{0}}$.
Hence $F_{\theta_{0}}[M]=J_{\theta_{0}}$. We show that this leads
to a contradiction. Define
\[
a_{x}=\langle0|M_{x}|0\rangle,\qquad b_{x}=\langle1|M_{x}|1\rangle,\qquad c_{x}=\langle2|M_{x}|2\rangle.
\]
We have
\[
p_{x}=\tr\rho M_{x}=\frac{a_{x}+b_{x}}{2},
\]
and
\[
\tr D_{1}M_{x}=\frac{a_{x}-b_{x}}{2}.
\]
Hence the $(1,1)$ component of the classical Fisher information matrix
induced by $M$ is
\[
F_{\theta_{0}}[M]_{11}=\sum_{x:p_{x}>0}\frac{(\tr D_{1}M_{x})^{2}}{p_{x}}=\sum_{x:p_{x}>0}\frac{(a_{x}-b_{x})^{2}}{2(a_{x}+b_{x})}.
\]
Since
\[
\sum_{x}a_{x}=\sum_{x}b_{x}=1,
\]
this can be written as
\begin{equation}
F_{\theta_{0}}[M]_{11}=1-\sum_{x:p_{x}>0}\frac{2a_{x}b_{x}}{a_{x}+b_{x}}.\label{eq:A7}
\end{equation}
Therefore,
\begin{equation}
F_{\theta_{0}}[M]_{11}=1\quad\Longrightarrow\quad a_{x}b_{x}=0\quad\text{for every }x.\label{eq:A8}
\end{equation}

We now show that (\ref{eq:A8}) forces
\begin{equation}
F_{\theta_{0}}[M]_{22}\leq\frac{1}{2}.\label{eq:A9}
\end{equation}
Suppose first that $b_{x}=0$ and $p_{x}>0$. Since $M_{x}\geq0$,
positivity implies
\[
M_{x}|1\rangle=0.
\]
Thus
\[
\langle2|M_{x}|+\rangle=\frac{1}{\sqrt{2}}\langle2|M_{x}|0\rangle,
\]
and consequently
\[
\tr D_{2}M_{x}=\frac{1}{2}\re\langle2|M_{x}|0\rangle.
\]
By positivity,
\[
|\langle2|M_{x}|0\rangle|^{2}\leq a_{x}c_{x}.
\]
Since $p_{x}=a_{x}/2$, we obtain
\begin{equation}
\frac{(\tr D_{2}M_{x})^{2}}{p_{x}}\leq\frac{c_{x}}{2}.\label{eq:A10}
\end{equation}
The same argument, with $|0\rangle$ and $|1\rangle$ interchanged,
applies when $a_{x}=0$ and $p_{x}>0$. Hence (\ref{eq:A8}) and (\ref{eq:A10})
give
\[
F_{\theta_{0}}[M]_{22}\leq\frac{1}{2}\sum_{x}c_{x}=\frac{1}{2}.
\]
This contradicts
\[
F_{\theta_{0}}[M]=J_{\theta_{0}}=I_{2}.
\]
Therefore, equality in (\ref{eq:A_SLDbound}) cannot be attained by
any locally unbiased estimator.


\begin{thebibliography}{10}
\bibitem{kimizu}M. Kimizu, F. Tanaka, and A. Fujiwara, ``Adaptive
quantum state estimation for two optical point sources,'' Phys. Rev.
A \textbf{109}, 032434 (2024).

\bibitem{qest}J. Zhang and J. Suzuki, ``QestOptPOVM: An iterative
algorithm to find optimal measurements for quantum parameter estimation,''
arXiv:2403.20131.

\bibitem{holevo}A. S. Holevo, ``Probabilistic and Statistical Aspects
of Quantum Theory,'' (2nd English edition) (Edizioni della Normale,
Pisa., 2011).

\bibitem{helstrom}C.W. Helstrom, Quantum Detection and Estimation
Theory (Academic Press, New York, 1976).

\bibitem{matsumoto_pure}K. Matsumoto, \textquotedblleft A new approach
to the Cram\'{e}r-Rao-type bound of the pure-state model,\textquotedblright{}
J.Phys. A 35 3111-3123. MR1913859 (2002).

\bibitem{gill_massar}R. D. Gill and S. Massar, ``State estimation
for large ensembles,'' \textit{Physical Review A}, \textbf{61}, 042312
(2000). 

\bibitem{yama_tomo}K. Yamagata, ``Efficiency of quantum state tomography
for qubits,'' Int. J. Quant. Inform., 9, 1167 (2011). 

\bibitem{conic}M. Hayashi and Y. Ouyang, ``Tight Cram\'{e}r-Rao
type bounds for multiparameter quantum metrology through conic programming,''
Quantum 7, 1094 (2023).

\bibitem{adaptive}A. Fujiwara, ``Strong consistency and asymptotic
efficiency for adaptive quantum estimation problems,'' Journal of
Physics A: Mathematical and Theoretical 44(7) 079501-079501 (2011). 

\bibitem{Chiribella}G. Chiribella, G. M. D'Ariano, and D. Schlingemann,
``How continuous quantum measurements in finite dimensions are actually
discrete,'' Phys. Rev. Lett. \textbf{98}, 190403 (2007).

\bibitem{sufficient}D. Petz, ``Sufficient subalgebras and the relative
entropy of states of a von Neumann algebra,'' Communications in Mathematical
Physics 105, 123\textendash 131 (1986). 

\bibitem{jordan-yamagata}K. Yamagata, ``Quantum Sufficiency for
Self-Adjoint Statistical Models via Likelihood-Type Operators on Real
{*}-Subalgebras and Real Jordan Algebras,'' arXiv:2604.23292 (2026).

\end{thebibliography}
\end{document}